\documentclass[11pt,letter]{article}
\pdfoutput=1

\usepackage{amsmath,amsfonts,amsthm,enumitem,graphicx}
\usepackage{microtype}
\usepackage{amssymb}
\usepackage{booktabs}
\usepackage[hidelinks]{hyperref}
\usepackage{fullpage}
\usepackage{multicol}
\usepackage{times}
\usepackage{todonotes}
\usepackage{xspace}
\usepackage{authblk}
\usepackage{cite}
\usepackage[T1]{fontenc}

\newtheorem{theorem}{Theorem}[section]
\newtheorem{lemma}[theorem]{Lemma}
\newtheorem{corollary}[theorem]{Corollary}
\newtheorem{observation}[theorem]{Observation}

\newtheorem{proposition}[theorem]{Proposition}

\newcommand{\Z}{\mathbb{Z}}
\newcommand{\R}{\mathbb{R}}
\newcommand{\Spec}{\mathit{Spec}}
\newcommand{\Bohr}{\mathit{Bohr}}
\newcommand{\tcolor}{\mathit{color}}
\newcommand{\idle}{\mathrm{idle}}
\newcommand{\diam}{\mathrm{diam}}
\newcommand{\pre}{\mathrm{pred}}
\newcommand{\suc}{\mathrm{succ}}
\newcommand{\Real}{\mathrm{Re}}

\newcommand{\bigo}{\mathcal{O}}
\newcommand{\divides}{\!\bigm|\!}
\newcommand{\mult}{\!\cdot\!}
\newcommand{\ka}{\kappa}
\newcommand{\RR}{{\textsc{RR}}\xspace}
\newcommand{\Mod}{\ \mathrm{mod}\ }

\graphicspath{{./fig/}}

\title{\bfseries Ergodic Effects in Token Circulation}

\date{}

\author[1]{Adrian Kosowski}
\affil[1]{Inria Paris and IRIF, Universit\'e Paris Diderot, France}
\author[2]{Przemys\l{}aw Uzna\'nski}
\affil[2]{Department of Computer Science, ETH Z\"{u}rich, Switzerland}

\begin{document}

\maketitle

\begin{abstract}

We consider a dynamical process in a network which distributes all particles (tokens) located at a node among its neighbors, in a round-robin manner.

We show that in the recurrent state of this dynamics (i.e., disregarding a polynomially long initialization phase of the system), the number of particles located on a given edge, averaged over an interval of time, is tightly concentrated around the average particle density in the system. Formally, for a system of $k$ particles in a graph of $m$ edges, during any interval of length $T$, this time-averaged value is $k/m \pm \widetilde\bigo(1/T)$, whenever $\gcd(m,k) = \widetilde\bigo(1)$ (and so, e.g., whenever $m$ is a prime number). To achieve these bounds, we link the behavior of the studied dynamics to ergodic properties of traversals based on Eulerian circuits on a symmetric directed graph. These results are proved through sum set methods and are likely to be of independent interest.

As a corollary, we also obtain bounds on the \emph{idleness} of the studied dynamics, i.e., on the longest possible time between two consecutive appearances of a token on an edge, taken over all edges. Designing trajectories for $k$ tokens in a way which minimizes idleness is fundamental to the study of the patrolling problem in networks. Our results immediately imply a bound of $\widetilde \bigo(m/k)$ on the idleness of the studied process, showing that it is a distributed $\widetilde \bigo(1)$-competitive solution to the patrolling task, for all of the covered cases. Our work also provides some further insights that may be interesting in load-balancing applications.

\end{abstract}



\section{Introduction}

The notion of \emph{ergodicity} captures the characteristic of a dynamical system which displays averaging behavior in both its state space and in time. For example, a finite Markov chain is ergodic if it is both aperiodic and irreducible, that is, for any sufficiently large integer $t$, it has positive probability to go between any pair of states in exactly $t$ steps. In this work, we focus on systems consisting of multiple particles (or tokens) moving around sites in a finite graph. In such a setting, an important manifestation of ergodicity is the property of \emph{time-averaging}: a measurement performed at a fixed location in the system over a period of time is representative of the averaged behavior of the system, for any site or location.

As an example, consider a system of $k$ non-interacting random walk particles on a connected non-bipartite symmetric directed graph, moving around the graph in synchronous steps. A direct consequence of the ergodicity of this system is that, in the limit, any multiset of $k$ arcs has the same probability of being traversed by the $k$ walks in any given step. The slightly weaker property of time-averaging for this system is phrased as follows: if the graph has $m$ arcs, then for any given arc, the average number of particles traversing this arc tends to $k/m$ per step, when averaging over sufficiently long time intervals.
However, when the number of particles $k$ in the system is small, the length of the time window required to obtain a reliable time-averaged measurement at any single measurement point becomes long:
if we observe any given arc for $\Theta(m^2/k^2)$ consecutive steps, we have constant probability of \emph{not} spotting any particle on the arc, even when the expected number of particles passing through this arc during the considered time interval is $\Theta(m/k)$ (cf.\ e.g.~\cite{ES11,KKPS13} for related analyses and more general considerations of parallel random walks).

In this paper, we investigate a way to obtain faster time-averaging behavior in a system with $k$ particles, and look at the algorithmic consequences of such a property. We study the fundamental deterministic process of particle propagation in which $k$ tokens circulate in a symmetric directed graph, with each node propagating all the tokens it contains to its neighbors by a round-robin mechanism (\emph{RR dynamics}), in each synchronous round. This process is closely related to (and in fact, a variant of) chip-firing games in the sandpile model (see Section~\ref{sec:introrr} for more details), but has the advantage that it displays meaningful behavior even in the regime of a very small number of tokens. At the same time, it bears numerous similarities to the model of non-interacting parallel random walks. These similarities have been studied and exhibited in contexts ranging from parallel cover time~\cite{DKPU14,KKPS13,KosowskiP14} to load-balancing properties~\cite{DF09,AB13,KosowskiP14,BKKMU15,DBLP:conf/analco/ShiragaYKY16}.

As our main result, we show that under the assumption of co-prime values of $m$ and $k$, the \RR dynamics in its recurrent state displays very fast time-averaging on arcs: in a time interval of length $T$, the number of tokens passing through any arc is always in the range $kT/m \pm \widetilde\bigo(1)$, i.e., the average number of tokens observed per step during the interval is $k/m \pm \widetilde\bigo(1/T)$, where the $\widetilde\bigo$ notation hides poly-logarithmic factors in $m$. In particular, in the studied scenario, each arc is visited by a token at least once in $\widetilde\bigo(m/k)$ steps. This effect of quasi-regular traversal of arcs in the \RR dynamics' recurrent state is only observed when the considered graph is symmetric. To prove its existence, we build on links between the \RR dynamics and the structure of Eulerian circuits in the graph. Our main contribution is to define and analyze a process of moving around the nodes of a graph in time, by performing traversals of fragments (contiguous sequences of arcs) of a given Eulerian circuit, and occasionally switching between different fragments of the circuit which meet at the same vertex. Such a process is shown to have the property that any pair of temporally and spatially separated states can be reached from one another within at most $\widetilde\Theta (1)$ operations of switching between fragments of the Eulerian circuit. The latter result is obtained by proving, in a framework of additive combinatorics, that iterated sums of the set of lengths of closed sub-circuits of any Eulerian circuit grow rapidly, covering the set of all integers in a small number of iterations.

The rest of the paper is organized as follows. We continue the introduction by describing the context of the \RR dynamics in Section~\ref{sec:introrr}, and some applications of our results in Section~\ref{sec:motivation}. We then provide an exposition of the technical results in Section~\ref{sec:techresults}. The remaining sections are devoted to a proof of the results, starting with an exposition of the main structural lemmas on mixing properties of circulations in Section~\ref{sec:sumset} and~\ref{sec:circulations}. We relate these results to \RR dynamics in Section~\ref{sec:rr}.


\vspace{-2mm}
\subsection{Overview of \RR Dynamics}\label{sec:introrr}
\vspace{-1mm}

The \RR dynamics, variously referred to as the \emph{round-robin} mechanism, the \emph{rotor-router}, the \emph{Eulerian walker} model, the \emph{ant walk} model, or the \emph{Propp machine}, describes a deterministic token propagation mechanism on a (symmetric) directed graph.\footnote{A process defined on a symmetric directed graph may also be considered as running on the corresponding undirected graph. We consider the symmetric directed variant for an easier discussion of (Eulerian) circuits.} In each step, the state of the system is represented by the locations of the $k$ tokens on nodes of the graph, and a pointer for each node, indicating one outgoing arc. In discrete, synchronous steps, all tokens are propagated according to the deterministic round robin rule, where after sending out each token, the node sending the token advances its pointer to the next outgoing arc in some fixed cyclic ordering of its out-neighborhood.

The \RR dynamics was first considered in~\cite{PhysRevLett.77.5079} and is a close cousin of the dynamics of \emph{sandpiles}. The latter model was introduced 30 years ago in the seminal work~\cite{BTW87} (cf. e.g.~\cite{RS17} for an overview of its algorithmic properties) and gave rise to the notion of \emph{self-organized criticality} (SOC) --- a concept since applied in dozens of fields, ranging from Earth sciences to neurobiology, to model systems whose evolution occurs close to a critical point. Whereas sandpile dynamics were the first synthetic physical system recognized to display SOC behavior, the \RR dynamics has been shown to exhibit comparable effects, visible for instance in the almost identical long-distance correlation patterns of the two processes on the grid, cf.\ e.g.~\cite{PhysRevLett.77.5079,DD14}. In the large-$k$ regime, the RR system and the synchronous sandpile model (with chip-firing rules~\cite{BLS91,KNTC94}) are comparable in many other respects. For instance, in a network load-balancing context, they both belong to a class of processes with a highly desirable property of cumulative load balancing over arcs outgoing from each node~\cite{BKKMU15,CS06,DF09}. In the small-$k$ regime, the RR dynamics has been studied in the context of its limit behavior and has often been compared to parallel random walk processes, e.g., in terms of its cover time on graphs~\cite{DBLP:journals/siamcomp/AfekG94,DKPU14,Frae70,GR,KosowskiP14,YWB03}.

The RR dynamics, like any other finite deterministic system, exhibits an initial transient phase, after which it stabilizes to a cyclic traversal of a set of recurrent states. For $k=1$ token, multiple authors have observed that on symmetric graphs, the dynamics stabilizes in a polynomial number of steps to a traversal of an Eulerian circuit~\cite{PhysRevLett.77.5079,DBLP:journals/siamcomp/AfekG94,YWB03} (for the non-symmetric and infinite settings, see e.g.~\cite{HP10}). The underlying link between the RR dynamics and Eulerian circuits is also displayed by the relation between the number of possible states of the system and the number of Eulerian circuits in the graph, as given by the BEST theorem (cf.\ e.g.~\cite{BGHIKK09} for further discussion). For $k>1$, a structural characterization of the limit behavior of \RR for an \emph{arbitrary} number $k>1$ of tokens was shown in \cite{CDGKLU15}. The obtained characterization shows that the \RR process provides a way of organizing tokens into balanced subsets, each of which follows a well-defined walk in some closed circuit on a subset of arcs of the graph. When the graph is \emph{not symmetric}, such token subsets may be arbitrary and display little regularity: for example, on a one-directional cycle in which each node has out-degree $1$, a group of $k$ tokens may perform a cyclic traversal of the graph while always located on the same node. Herein we show how, and (plausibly) why, in any \emph{symmetric} directed graph the RR dynamics will (under almost all parameter configurations) self-organize into recurrent states with an arrangement of tokens in which distances between tokens on their respective circuits are almost uniform.

\vspace{-2mm}
\subsection{Motivation and Consequences of Results}\label{sec:motivation}
\vspace{-1mm}

In Section~\ref{sec:techresults} we provide a formal exposition of the obtained structural results for Eulerian circuits, and their consequences for the RR model. Before we do this, we present here two examples of concrete scenarios in which the results of this paper find direct application. The first concerns the deployment of the RR process in the task of regularly patrolling network edges, while the second makes use of the RR rules to balance load in a network fairly over periods of time. We remark that for neither of these tasks were decentralized solutions of comparable quality previously known, whether deterministic or randomized. In particular, as discussed further, the application of random-walk-based techniques fails completely for both tasks.

\subsubsection{Network Patrolling}

The term \emph{patrolling} refers to an act of surveillance, which involves walking perpetually around an area in order to protect or supervise it. It is a convenient description of  numerous settings, such as locating objects or humans that need to be rescued from a disaster, ecological monitoring, or intrusion detection. Network
administrators may patrol network nodes and links to detect network failures or to discover web pages which need to be indexed by search engines, cf.\ \cite{MRZD02}.
Patrolling has been recently intensively studied  in robotics  (cf.\ \cite{ARSTMCC04,C04,EAK09,ESK08,HK08,MRZD02,YWB03})
where it is often viewed as a version of terrain {\em coverage}, a central task in robotics. Patrolling boundaries of terrains and their area have been studied in \cite{AKK08,EAK09,ESK08,PFB10} with
approaches placing more emphasis on experimental results.

The accepted measure of the algorithmic efficiency of
patrolling is called  {\em idleness} or {\em refresh time} and describes  the frequency with
which the points of the environment are visited
(cf. \cite{ARSTMCC04,C04,EAK09,ESK08,MRZD02,PFB10});
this criterion was first introduced in \cite{MRZD02}.
Depending on the
requirements, the idleness may sometimes be viewed as the
average (\cite{EAK09}), worst-case (\cite{BGHIKK09,YWB03}),
probabilistic (\cite{AKK08}) or experimentally verified
(\cite{MRZD02}) time elapsed since the last visit to a node or edge (cf. also \cite{ARSTMCC04,C04}).  A survey of diverse approaches to patrolling based on idleness criteria can be found in \cite{ARSTMCC04}, whereas a theoretical analysis of approaches to patrolling in graph-based models can be found in \cite{C04}.

\paragraph{Contribution:} In the standard edge-patrolling setting on graphs, the asymptotically optimal value of (worst-case) idleness which can be achieved using a team of $k$ agents on a graph with $m$ edges is $\Theta(m/k)$. This is achieved, e.g., by deploying $k$ agents in a centralized manner with equal spacing, traversing a predefined Eulerian circuit in the graph (cf.~e.g.~\cite{C04}). We show that the RR dynamics provides a completely decentralized (and self-organized) solution to the edge-patrolling problem on graphs, achieving idleness of $\widetilde\Theta(m/k)$ in a wide variety of settings, in particular, when $\gcd(k,m)=1$, as well as for trees. (We remark that some empirical studies of patrolling of trees with the RR process in trees were provided in~\cite{CKKT16}, without theoretical analysis.) In Table~\ref{tab:results} we present the obtained results on the \emph{patrolling competitive ratio} of the RR process, i.e., the achieved value of the idleness divided by the optimal $\Theta(m/k)$ idleness of the centralized process. To the best of our knowledge, designing a decentralized patrolling strategy with small idleness was a relevant open question in the area, and the RR approach may prove to be a viable practical approach. As pointed out in the discussion of time-averaging processes at the beginning of the paper, even when considering probabilistic measures of idleness and when the considered graph is a cycle, a team of parallel random walks achieve an idleness of $\Omega(m^2/k^2)$, which corresponds to a competitive ratio of $\Omega(m/k)$.

\begin{table*}[t!]
\centering

  \label{tab:networks}
\begin{tabular}{llll}
\toprule
\textbf{Scenarios for RR dynamics} & \emph{Patrolling competitive ratio} & \emph{Cumulated load discrepancy} & \emph{Reference} \\
\midrule
small $\gcd(k,m)$: & $\widetilde\bigo( \gcd(k,m))$ & $\widetilde\bigo(\gcd(k,m))$ &  Pro. \ref{pro:idleness}, \ref{pro:cumul}\\
\midrule
large $k$: & $\bigo(1)$, for $k \ge \frac{3}{4}m$ & $-$ & Pro. \ref{pro:idleness}\\
small $k$: & $\widetilde\bigo(\sqrt{k})$ & $-$ & Pro. \ref{pro:idleness}\\
small $\diam$: & $\bigo(\diam)$& $\bigo(\diam)$ &  Pro. \ref{pro:idleness}, \ref{pro:cumul}\\
small $n$: & $\widetilde\bigo( \sqrt{n})$  & $\widetilde\bigo(\sqrt{n})$ &  Pro. \ref{pro:idleness}, \ref{pro:cumul}\\
tree topology: & $\bigo(1)$, for trees &  $-$ & Pro. \ref{pro:idleness}\\[2mm]
\midrule
\multicolumn{2}{l}{\textbf{Comparison of strategies (for $m$ prime)}}\\
\midrule
RR dynamics: & $\widetilde\bigo(1)$ & $\widetilde\bigo(1)$ & \\[2mm]
$k$ parallel random walks & $\Omega(m/k)$ & $\Omega(\sqrt T)$ in expectation  & \\[-1mm]
or randomized rounding~\cite{FS09}: &  (also in expected sense) &  (i.e., unbounded) &  \\
\bottomrule
\end{tabular}

\caption{Summary of obtained performance bounds on RR dynamics. The underlying graph is assumed to be symmetric directed, on $n$ nodes, $m$ arcs, diameter $\diam$, and with $k$ tokens. The cumulated load discrepancy presented in the table holds eventually over any time interval, of arbitrary length $T$.}
\label{tab:results}
\end{table*}

\subsubsection{Load Balancing}

Process in which nodes exchange tokens have also been extensively studied from the perspective of \emph{load balancing}, where one sees the number of tokens at a node as its load, and aims at minimizing the \emph{discrepancy}, i.e., the difference between the maximal and minimal load of a pair of nodes of a network.

Many distributed load balancing processes are inspired by continuous dynamics from nature. For example, the \emph{heat equation}, which describes real-world processes such as heat and particle diffusion, can also be used to model the following continuous load balancing process, assuming arbitrarily divisible load instead of indivisible tokens. Each node with load $x(u)$ sends to each of its $d$ neighbours load $x(u)/(d+1)$, and keeps $x(u)/(d+1)$ load to itself. In such \emph{diffusive load balancing}, it is well known that loads in $d$-regular graphs will eventually become perfectly balanced~\cite{SS94}. However, such balancing might be impossible to achieve in a discrete setting, and in general it is impossible to simulate a continuous process in a discrete one. Discrete load balancing schemes, which in principle deal with integer-rounding issues in a continuous process,  are also in general much harder to analyze. In \cite{RSW98}, a process that sends $\lfloor x(u)/(d+1) \rfloor$ or $\lceil x(u)/(d+1) \rceil$ load over each edge is shown to achieve vertex discrepancy of $\bigo(d \log n / \mu)$ after $\bigo(\log(K n)/\mu)$ steps for $d$-regular graphs, where $K$ is the initial load discrepancy and $\mu$ is the spectral gap of the transition matrix of the underlying Markov chain.  Since \cite{RSW98}, many variants of discretization of diffusive process have been proposed, see \cite{ABS12,FGS10,FS09,SS12,BKKMU15}.

\paragraph{Contribution:} This paper opens the ground for the analysis of a more fine-grained measure in discrete load balancing, namely, balancing the load passing through a node, cumulated (summed) over intervals of time. Thus, for an arbitrary fixed interval of length $T$, rather than consider the load $x_t(u)$ of a node $u$ at time $t$, we look at the \emph{cumulated load} $\sum_{\tau=t}^{t+T-1} x_\tau(u)$, and bound the \emph{discrepancy} of this value from its average over all nodes of the graph. The existence of this measure was previously indicated in our work~\cite{BKKMU15}, though we managed to obtain relatively weak bounds for a wider family of processes. Here, we show that for the specific load balancing process governed by the RR process, the cumulated discrepancy is very low: the cumulated load of every node is identical up to $\widetilde\bigo(1)$ when $\gcd(k,m)=1$ (see Table~\ref{tab:results} for a detailed overview of results). In other words, the total number of operations, such as storing a task, sending or receiving a task from a neighbor, which are performed by every node, is fairly balanced over all nodes of the network, counting over any window of time. Note that the cumulative load balancing problem cannot be addressed using a random-walk-type approach: indeed, even a process propagating a single superfluous token with a random walk among otherwise balanced nodes can cause significant discrepancy of cumulated load, which eventually becomes unbounded. This is, in particular, the case for the elegant randomized rounding approach of Friedrich and Sauerwald~\cite{FS09}, which is extremely efficient at bounding load differences at any fixed moment of time, but which allows the differences of cumulated load of nodes over a time interval of length $T$ to increase (in expectation) proportionally to $\sqrt T$. (Such behavior is observed even on simple graph topologies, such as the three-vertex cycle, since intuitively, the best guarantees for such a randomized process follow directly from the law of large numbers. The guarantees we obtain for the RR dynamics have a more subtle explanation which relies on combinatorial and number-theoretic properties. Other deterministic load balancing processes known to us show no similar advantage for cumulated load balancing).

We note that our results provide insights for the behavior of the RR dynamics after it has stabilized, and have no direct bearing on its initial stabilization phase.

\section{Exposition of Technical Results}\label{sec:techresults}

\subsection{Notation}

In this paper, we will consider graph $G=(V,E)$ to be directed and symmetric, i.e., such that every arc $e = (u,v)$ has a corresponding opposite arc $(v,u) =: -e$. We denote $|E| = m$ and $|V| = n$. We allow the graph to have multi-arcs as well as self-loops, however, to extend our notation to loops, we say that a loop is the opposite of itself, only (even if there are many loops in a single vertex). For an arc $e = (u,v)$, we denote $u = \pre(e)$ and $v = \suc(e)$. Unless otherwise stated, we will assume that $G$ is not bipartite, to avoid obvious issues with the periodicity of $2$ in the dynamics on $G$.

We denote the set of integers by $\Z$, the set of non-negative integers by $\Z_+$, and the set of integers in modulo-$\eta$ arithmetic as $\Z_\eta$. For $A,B \subseteq Z_\eta$ their \emph{sumset} is $A+B := \{(a+b) \mod \eta : a \in A, b \in B\}$, and for $k \in \Z_+$, $kA := A +  \ldots + A$, with the sum spanning over $k$ copies of $A$.

\subsection{Eulerian Mixing in Time and Space}\label{sec:exposition}

The main technical ingredients of this paper are linked to properties of traversals of Eulerian circuits, and, more broadly, of circulations on symmetric directed graphs. We formulate them in this section in the most intuitive form, in a way which is independent of the \RR process.

We will call $\varphi : E \rightarrow E$ a (unit) \emph{circulation} on the set of arcs of a symmetric directed graph $(V, E)$ if $\varphi$ is a bijection on $E$ such that for all arcs $e\in E$, $\suc(e) = \pre(\varphi(e))$.

For a pair of vertices $u, v \in V$ and $t\in \Z_+$, we will write $u \to_t v$ if there exists an edge $e$ with $\pre(e) = u$ such that $\suc(\varphi^t(e)) = v$, i.e., it is possible to reach $v$ from $u$ by following a $t$-step walk along the circulation. We will also write $v \to_{-t} u$ whenever $u \to_{t} v$, for $t\in \Z_+$.

We will call $\varphi$ \emph{Eulerian} if the orbit of any arc is an Eulerian circuit on $(V,E)$, i.e., for any $e\in E$, we have $\{\varphi^t(e) : t\in \Z\} = E$. In general, we will consider our circulation $\varphi$ as a union of $g\geq 1$ orbits (cycles), denoted $E_1, \ldots, E_g$, forming a partition of the arc set.

For a circulation $\varphi$ on a graph $G=(V, E)$, let $G_{\varphi}$ be the graph with vertex set $V \times \Z$ and edge set $\{\{(u,t_u), (v, t_v)\} : u \to_{(t_v - t_u)} v\}$. For the important case of Eulerian circulations, i.e., $g=1$, the (simpler) version of the main technical lemma of our paper may be stated as follows.

\begin{lemma}\label{lem:maineuler}
Let $\varphi$ be an Eulerian circulation on a non-bipartite symmetric directed graph $G$. Then, graph $G_\varphi$ has diameter $\bigo(\log^2 n)$.
\end{lemma}

The assumption that the considered graph is symmetric cannot be omitted in the statement of Lemma~\ref{lem:maineuler}; an example of a non-symmetric graph for which $G_\varphi$ has larger diameter is presented in Fig.~\ref{fig:dwucykle}.

\begin{figure}
\centering\includegraphics[scale=1.25]{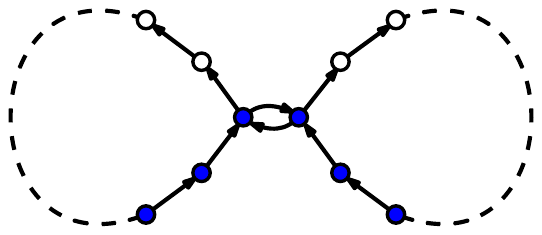}
\caption{A non-symmetric directed graph for which Lemma~\ref{lem:maineuler} does not hold: the unique Eulerian circulation $\varphi$ covering $G$ results in $G_{\varphi}$ that has diameter linear in $n$. If we consider a configuration of RR dynamics with one token per each blue vertex, symmetrically on both sides, the configuration repeats every $m/2$ steps, and no mixing occurs.}
\label{fig:dwucykle}
\end{figure}

As a side remark, we note that the above Lemma has an intuitive interpretation in that it bounds the distance between a pair of time-separated points $(u,0)$ and $(u,l)$ by $\bigo(\log^2 n)$. We can thus write the following informal corollary.
\begin{corollary}
For any fixed Eulerian circuit in a non-bipartite symmetric directed graph $G$ and any $l \in \Z_{m}$, there exists a closed walk over the arcs of $G$ of length congruent to $l$ modulo $m$, constructed by a successive traversal of $\bigo(\log^2 n)$ contiguous fragments of the given Eulerian circuit.
\end{corollary}
Note that the trajectory between nodes $(u,0)$ and $(u,l)$ which follows from a shortest path in graph $G_\varphi$ may contain ``positive'' traversals along iterated applications of $\varphi$, as well as ``negative'' traversals along iterated applications of $\varphi^{-1}$. However, one can ensure that all arcs appear in the considered walk a non-negative (or positive) number of times by augmenting it with a traversal of sufficiently large number of iterations of the complete Eulerian circuit $\varphi^m$. This does not change the length of the walk, modulo $m$.

The proof of Lemma~\ref{lem:maingeneral} relies on a characterization of the set of lengths of \emph{self-intersections} of an Eulerian circuit. Intuitively, $t$ belongs to the set of self-intersections if there exists a closed walk $u \to_t u$ along the Eulerian circuit. In Section~\ref{sec:sumset} we exploit the spectral properties of the set of self-intersections to show that for given $\eta$, any value in $\Z_\eta$ may be obtained as a sum of at most $\bigo(\log^2 \eta)$ lengths of self-intersections. The discussion in the section is formulated in slightly more general notation, leading to Lemma~\ref{lem:explosion}, which will also prove useful in the consideration of non-Eulerian circulations. The claim then follows readily (Section~\ref{sec:circulations}).

Extending considerations beyond Eulerian circulations, we obtain the following generalization of Lemma~\ref{lem:maineuler} to the case of $g>1$.




\begin{lemma}\label{lem:maingeneral}
Let $\varphi$ be a unit circulation on a non-bipartite symmetric directed graph $G$. Then, $G_{\varphi}$ has diameter at most $\bigo(g \log^2 n)$.
\end{lemma}

The proof of this lemma relies on the same ingredients (Section~\ref{sec:sumset}) as the considerations in the Eulerian case, however, the sum set technique only provides us with a way of showing mixing in time when considering time values modulo the greatest common divisor of the lengths of all cycles in the circulation. The general result for arbitrary integer time values follows from an application of the Chinese Remainder Theorem (Section~\ref{sec:circulations}).

The obtained bounds on the diameter of $G_{\varphi}$ are tight up to polylogarithmic factors, and we subsequently apply them to the analysis of the RR model.

We close this section with the following side remark, which appears to be interesting in its own right but has no bearing on the main results of the paper. It is possible to obtain a slightly stronger spectral property of $G_{\varphi}$ for Eulerian $\varphi$ than merely a polylogarithmic bound on the diameter. For any Eulerian circuit in $G$, one can design a random process which will terminate at a well-mixed node in $G$ at a well-mixed time, by traversing a polylogarithmic number of fragments of the Eulerian circuit.

\begin{corollary}
Let $\varphi$ be an Eulerian circulation on a non-bipartite symmetric directed graph $(V, E)$. Then, there exists an assignment of weights $w : E(G_\varphi) \to [0,1]$, such that the reversible Markov chain following $(G_{\varphi},w)$ achieves good mixing on $V \times \Z_{m}$ starting from any vertex from this set, after $\widetilde\bigo(1)$ steps.
\end{corollary}

The proof follows from a straightforward adaptation of the spectral arguments used in our analysis in Section~\ref{sec:sumset}.

\subsection{RR Dynamics and Short-Time-Averaging}\label{sec:rrintro}

We start this section by formally defining the \RR dynamics. Each vertex $v$ of $G$ is equipped with a fixed ordering of all its outgoing arcs $\rho_v = (e_1,e_2, \ldots, e_{\text{deg}(v)})$.
A~\emph{state} $x$ at the current time step  $t$ is a tuple:
$x_t = ((\mathit{pointer}_v)_{v \in V},(L(v))_{v \in V}),$
where $\mathit{pointer}_v$ is an arc outgoing from node $v$, which is referred to as \emph{the current port pointer at node} $v$, and $L(v)$ is the number of tokens at any given node. The state space of the RR dynamics is denoted $X$, and its cardinality is in general readily seen to be (exponentially) larger than $|V|$. The system is initialized to an arbitrary state of the state space $X$. For an arc $(v,u)$, let $\mathit{next}((v,u))$ denote the arc after the arc $(v,u)$ in the cyclic order $\rho_v$.
During each step, each node $v$ distributes in round-robin fashion all of its tokens, using the following algorithm:
\\[2mm]
\textbf{While} there is a token at node $v$, \textbf{do}:
\begin{enumerate}
\item Send the token to $\mathit{pointer}_v$,
\item Set $\mathit{pointer}_v = \mathit{next}(\mathit{pointer}_v)$.
\end{enumerate}
Note that during a single time step all tokens at a node $v$ are sent out and at exactly the next time step all those tokens arrive at their respective destination nodes.  The total number of tokens in graph is denoted as $k$.


When the system is in state $x \in X$, let $(L_x(e) : e\in E) \in \Z_+^m$ be the distribution vector of tokens over arcs of the RR dynamics during its transition from state $x$ to the subsequent state, i.e., let $L_x (e)$ represent the number of particles sent out along arc $e \in E$ in the current time step, starting from state $x$.

Once the RR dynamics has entered a recurrent state, the trajectory followed by the dynamics involves only a small fraction of its entire state space $X$, and as such, when viewed in high dimension, it is definitely not an ergodic process. However, we show that it is possible to make use of the time-mixing property of (almost) Eulerian circulations captured by Lemma~\ref{lem:maingeneral} to show that RR achieves the sought time-averaging behavior when considering only one-dimensional probe functions $L(e)$. Intuitively, in the recurrent state of the RR dynamics, individual tokens may be seen as traversing cycles of some unit circulation, with balancing behavior occurring at the meeting points of different cycles in the circulation.

\begin{theorem}
\label{thm:rr}
For any $e \in E$ and $t_0, T \in \Z_+$, such that $x(t_0)$ is recurrent, the RR dynamics satisfies:
$$
\frac{1}{T}\sum_{t = t_0} ^{t_0+T} L_{x(t)}(e) = \frac{k}{m} \pm \widetilde\bigo\left(\frac{gcd(k, m)}{T}\right).
$$
\end{theorem}
The above Theorem, which we prove in Section~\ref{sec:rr} after introducing more background properties of the \RR process, has immediate consequences to problems of idleness in patrolling  and cumulative load balancing. This is also discussed in Section~\ref{sec:rr}.

\section{Sum-sets of Self-Intersections on a Circulation}\label{sec:sumset}

Compared to the more high-level exposition in Section~\ref{sec:exposition}, in the subsequent analysis it is more convenient to work with measures of distance between pairs of arcs rather than pairs of vertices. We thus introduce all further notation based on the arc-based perspective.

\paragraph{Intersections:}
We will say that an assignment  $\lambda : E \rightarrow \Z_{\eta}$ along the cycle (cycles) of circulation $\varphi$ is a $\eta$-\emph{labeling} if the following condition is fulfilled: $\forall_{e\in E}\ \lambda(\varphi(e)) = \lambda(e)+1$. (This is only possible when $\eta \divides \gcd(|E_1|,|E_2|,\ldots,|E_g|)$.) We say that two arcs $e_1$ and $e_2$ outgoing from the same vertex, $\pre(e_1) = \pre(e_2)$,  span an \emph{intersection} of size $(\lambda(e_1)-\lambda(e_2)) \Mod \eta \in \Z_{\eta}$. We consider the \emph{set of all intersections} $A_{\lambda} \subseteq \Z_{\eta}$:  $$A_{\lambda} = \{ \lambda(e_1) - \lambda(e_2) : \pre(e_1) = \pre(e_2)\}.$$
Observe that $A_{\lambda}$ does not depend on the particular choice of $\lambda$, since all $\lambda$ are identical up to cyclical shift.
We immediately have $0 \in A_{\lambda}$ and $A_{\lambda} = -A_{\lambda}$. Observe also that even when $\eta$ is even, there is an odd element in $A_{\lambda}$. Otherwise, we could partition vertices into two sets, $\{\pre(e) : 2 \mid \lambda(e)\}$ and $\{\pre(e) : 2 \not| \lambda(e)\}$, which contradicts the assumption that $G$ is not bipartite.

Set $A_{\lambda}$ may have a non-trivial structure, e.g., may contain no small elements other than $0$ (see e.g.\ Fig.~\ref{fig:3kolory}). Fortunately, it turns out we can take advantage of its spectral properties (which we subsequently formally define).

\begin{figure}
\begin{minipage}{.33\textwidth}
\centering\raisebox{2.5cm}{$(a)\ $}\includegraphics[scale=1]{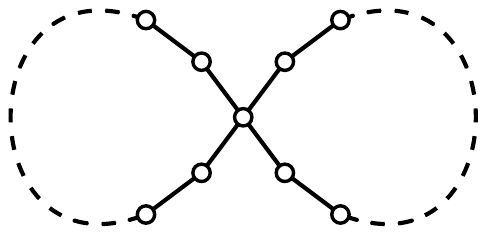}
\end{minipage}
\begin{minipage}{.33\textwidth}
\centering\raisebox{2.5cm}{$(b)\ $}\includegraphics[scale=.5]{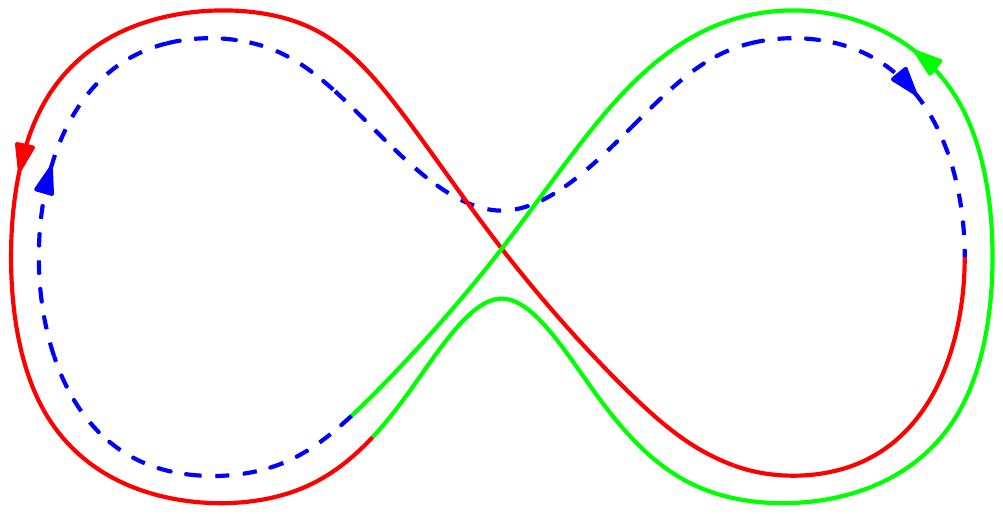}
\end{minipage}
\begin{minipage}{.33\textwidth}
\centering\raisebox{2.5cm}{$(c)\ $}\includegraphics[scale=.5]{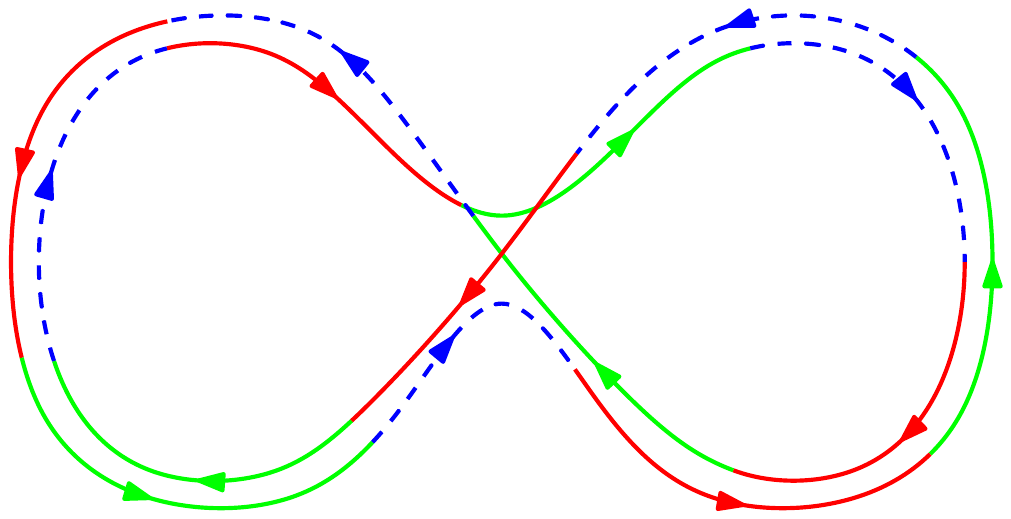}
\end{minipage}
\caption{(a) Example of a graph $G$ with a circulation with no ``small'' intersections: for any $m$-labeling, $A_{\lambda} \subseteq [m/4,3m/4]$. Examples of assignments of colors in $G$ for: (b) $m$-labeling and (c) $m/5$-labeling. In both cases, all 3 colors meet at least once.}
\label{fig:3kolory}
\end{figure}

\paragraph{Spectral properties of $A_{\lambda}$:}

We start by showing the following key property of the set $A_\lambda$ of all intersections, expressed by Lemma~\ref{lem:int_mod}. We then briefly recall the formalism of Bohr sets, and use Lemma~\ref{lem:int_mod} to characterize the structure of the Bohr set of $A_\lambda$.

\begin{lemma}
\label{lem:int_mod}
Let $\xi \in \Z_\eta$ with $1 \le \xi \le \eta/6$. Let $A_\lambda$ be the set of intersections for an arbitrary $\eta$-labeling of a circulation. Then, there exists $x \in A_\lambda$ such that $\xi \cdot x \in [\frac13 \eta, \frac23 \eta]$.
\end{lemma}
\begin{proof}
Let us consider a continuous version of set $E$, that is a metric space $\mathcal{E}$ constructed on the set $E \times [0,1)$, such that for every edge $e$, we set $\lim_{t \to 1} (e,t) = (\varphi(e),0)$. Observe that $\mathcal{E}$ is composed of disjoint spaces $\mathcal{E}_1,\mathcal{E}_2,\ldots,\mathcal{E}_g$, where each $\mathcal{E}_i$ is a continuous version of cycle $E_i$, isometric to $\R / (|E_i|\mult \Z)$ (i.e., to a cyclic interval of length $|E_i|$).

%

We say that elements of $\mathcal{E}$ of the form 
$((v,\cdot),0)$ are occurrences of $v$ in $\mathcal{E}$. Now we can extend the labeling $\lambda$ of $E$ to a labeling $\mathcal{L} : \mathcal{E} \rightarrow \R/(\eta\mult\Z)$, so that $\mathcal{L}((e,0)) = \lambda(e)$ and $\mathcal{L} \mid_{\mathcal{E}_i}$ is a locally distance-preserving mapping of $\mathcal{E}_i$ to $ \R / (\eta\mult \Z)$. ($\mathcal{L}$ is an ``interpolation'' of $\lambda$, increasing in a continuous fashion at a constant rate between discrete points on which $\lambda$ was defined.)

We will say that two points in $\mathcal{E}$ \emph{meet} if they are either  occurrences of the same vertex $v \in V$, or if they are of the form $(e,t)$, $(-e,1-t)$, for some $e\in E$ and $0<t<1$, where we recall that $e$ and $-e$ denote a pair of opposing arcs.

Now, for fixed $\xi \in \Z_{\eta}$ satisfying the assumption of the lemma, we assign to each point $x \in \mathcal{E}$ one of three colors  $\{0,1,2\}$ according to the following rule:
$$ \tcolor(x) = \begin{cases} 0 \text{\quad if } \xi \cdot \mathcal{L}(x) \in [0,\frac13\eta), \\ 1 \text{\quad if } \xi\cdot \mathcal{L}(x) \in [\frac13\eta,\frac23\eta), \\ 2 \text{\quad if } \xi \cdot \mathcal{L}(x) \in [\frac23\eta,\eta). \end{cases} $$
Thus, intuitively we have covered the continuous version of the arc set with closed directed continuous cycles, and each cycle is colored along its orientation with colors $0,1,2$, with colors successively applied to contiguous segments of length $\eta/(3\xi)$. See Figure~\ref{fig:3kolory} for example.
Now, by purely topological reasoning, we observe that there is a point of $\mathcal{E}$, where all 3 colors meet. To be more precise, one of the following happens:
\begin{enumerate}
\item There exist $x_0,x_1,x_2 \in \mathcal E$, all being occurrences of the same vertex from $V$, such that $\tcolor(x_i) = i$. It follows that $ \xi \cdot (\mathcal{L}(x_i) - \mathcal{L}(x_j)) \in [\frac13\eta,\frac23\eta]$ for some $i \not= j$.
\item There exist $x,y \in \mathcal{E}$ that meet and satisfy the following constraint (up to permutation of colors $0,1,2$): $\tcolor(x) = 0$, $\tcolor(x-\varepsilon) = 2$ (for arbitrarily small $\varepsilon$), and $color(y) = 1$.  It follows that $\xi \cdot
(\mathcal{L}(y) - \mathcal{L}(x)) \in [\frac13\eta,\frac23\eta)$.
\end{enumerate}
Since $\mathcal{L}$ and $\lambda$ agree on all the occurrences of the vertices (if we identify the values on the edges as taken on the starting points), and in the second case we can round pair $x,y$ to one of the endpoints of the edge they lie on, the claim follows.
\end{proof}

The above lemma, applied over different $\xi$, yields a valuable characterisation of the set of elements in $A_\lambda$.

\paragraph{Bohr set:}
We define the \emph{Bohr set} with radius $\alpha$ for a set $S \subseteq \Z_{\eta}$ as\footnote{For all undefined notation in additive combinatorics, we refer the reader to~\cite{Granville,TaoVu}.}:
$$
\Bohr (S,\alpha) = \{ \xi \in \Z_{\eta} :  \forall_{s\in S} \|s \cdot \xi \|_\eta \leq \alpha \},
$$
where $\|\cdot\|_\eta \in \R$ denotes the ``fractional part'' norm on $\Z_\eta$. (Formally, for $x \in \Z_\eta$, we put $\|x\|_\eta = \alpha$ for the unique real $\alpha \in (-1/2, 1/2]$ such that $x \equiv \alpha \eta \pmod \eta$; see also \cite{TaoVu}[beginning of Section 4.4].)

With this definition, Lemma~\ref{lem:int_mod} can be rephrased as follows: $\Bohr(A_\lambda, 1/3-\varepsilon) \cap (0,\eta/6] = \emptyset$, for arbitrarily small $\varepsilon>0$. This is not yet the spectral property we need, but $A_\lambda$ is easily transformable to a set with such a property. Since $A_\lambda = - A_\lambda$, we immediately have $\Bohr(A_\lambda, 1/3-\varepsilon) \cap [5/6 \eta, \eta) = \emptyset$ as well. We will say that $a$ is critical for $b$, if $\|a\cdot b\|_\eta \geq \alpha$. Since $x \cdot (2\cdot \xi) = (2 \cdot x) \cdot \xi$, if $x$ is a critical element for $2 \cdot \xi$, $2\cdot x$ is critical for $\xi$, thus we have $\Bohr(2 \mult A_\lambda, 1/3 - \varepsilon) \cap \big( [5/12 \eta, \eta/2) \cup (\eta/2, 7/12 \eta] \big) = \emptyset$. Additionally, we observe that for any $\xi \in (\eta/6, 5/12\eta]$, $\| 2 \cdot \xi \|_\eta \ge 1/6$, thus $\Bohr(A_\lambda, 1/6) \cap \big( (\eta/6, 5/12\eta] \cup [7/12\eta, 5/6 \eta) \big) = \emptyset$.

By $A_\lambda \cup 2\mult A_\lambda \subseteq A_\lambda+A_\lambda$, we conclude that $\Bohr(A_\lambda+A_\lambda,1/6) \subseteq \{0, \eta/2\}$. However, since we know that $A_\lambda$ contains an odd element, we can eliminate $\eta/2$ from the set. Thus finally, we obtain a spectral characterization of set $A_\lambda$.

\begin{lemma}
\label{cor:bohr_inter}
Let $A_\lambda$ be defined as in Lemma~\ref{lem:int_mod}. Then: $\Bohr(A_\lambda+A_\lambda, 1/6) = \{0\}$.
\end{lemma}

\paragraph{Additive combinatorics:}
Let $A \subseteq \Z_{\eta}$ be any set. We denote by $\hat A(j)$ the $j$-th ($j\in [0,\eta-1]$) Fourier coefficient in the transform of $1_A$, normalized so that $\hat A(0) = 1$ and $| \hat A(j)| \leq 1$, for all $j$. Formally, $\hat A(j) := \frac{1}{|A|}\sum_{a\in A} e(a j / \eta)$, where $e(x) \equiv \exp(2 \pi \mathbf{i} x)$).

Define the \emph{spectrum} $\Spec_\alpha(A)$ as the set of all $j \in [0,\eta-1]$ such that $|\hat A(j)|>\alpha$ (i.e., the set of indices of Fourier coefficients which are sufficiently large). The definition makes sense for $\alpha \in (0,1]$ and $0 \in \Spec_\alpha(A)$. (See~\cite{TaoVu}[Def. 4.34]). The following lemma about annihilation of the spectrum is folklore; we provide a short proof for completeness.

\begin{lemma}
\label{lem:spec1}
If $\Spec_{1 - 1/s}(A) = \{0\}$, then $\ka  A - \ka  A = \Z_{\eta}$, for $\ka \ge s \ln \eta$.
\end{lemma}
\begin{proof}
Let $\ka= s\ln \eta$, and let $r_{\ka A-\ka A}(l)$ denote the number of ways to represent the value $l \in \Z_{\eta}$ using  $\ka  A-\ka  A$. We have (cf. e.g.~\cite{Granville}[Sec I.8]):
\begin{align}
\label{eq:spectrum_collapse}
r_{\ka A-\ka A}(l)& =c_k \sum_j |\hat A (j) |^{2\ka} e(-jl/\eta) \geq \\
&\geq
c_k \left(1 - \sum_{j\neq 0} |\hat A (j) |^{2\ka}\right),\nonumber
\end{align}
where $c_k$ is some positive normalizing constant ($c_k = |A|^k/\eta$).

Next, for a set $A$ satisfying the assumption of the Lemma, we have for $j\neq 0$: $|\hat A (j) |^{2\ka} \leq (1-1/s)^{2s \ln \eta} < 1/\eta$. Hence, $\sum_{j\neq 0} |\hat A (j) |^{2\ka} < 1$. It follows that for all $l\in \Z_{\eta}$, we have $r_{\ka A-\ka A}(l) > 0$, and so $\ka  A - \ka  A = \Z_{\eta}$.
\end{proof}

In the following, we will also use the following property linking the Bohr set and the spectrum.

\begin{proposition}[\!\!\cite{TaoVu}, variant of Prop.\ 4.40]
\label{prop:tao}
Let $A,B \subseteq \Z_{\eta}$ be such that $A \subseteq B$ and $0 \in A$. Let $K \ge 1$ and $0 < \varepsilon < 1$. If $|B-A| \le K |B|$, then
$$A - A \subseteq \Bohr(\Spec_{1-\varepsilon}(B-A), \sqrt{8\varepsilon K}).$$
\end{proposition}
\noindent
For completeness, the proof of the above Proposition is provided in the Appendix.

We are now ready to show our final Lemma, where we extend considerations from the above Proposition to iterated sums of a set.
\begin{lemma}
\label{lem:spec_bohr}
For any set $A_1 \subseteq \Z_{\eta}$, such that $0 \in A_1$, we have that $\Spec_{1-1/576}(A_1 - \ka  A_1) \subseteq \Bohr(A_1,1/6)$, for some $\ka = \bigo(\log \eta)$.
\end{lemma}
\begin{proof}
Let us denote for short $B = \Bohr(A_1,1/6)$.
In the proof of the lemma, we will use an equivalent property characterizing $A_1$, which follows from symmetry of multiplication in expanding the Bohr set definition (i.e., $\forall {x \in  \Z_{\eta}, x \neq 0}\ :\ \exists_{\xi \in A_1} \|x \cdot \xi\|_\eta > 1/6$): 
\begin{equation}\label{eq:bohr}
\forall {X\subseteq \Z_{\eta}, X \not\subseteq B}\ :\ A_1 \not \subseteq \Bohr (X, 1/6).
\end{equation}
We now define sets $A_i$, starting from $A_2$, as follows: $A_{i+1} := A_i - A_1$. We will assume that $0 \in A_1$, so $0\in A_i$ for all $i$, and $A_{i+1} \supseteq A_i \supseteq \ldots \supseteq A_1$.

We now apply Proposition~\ref{prop:tao} to set $A_{i+1}$, with $K=2$ and $\varepsilon = 1/576$. From the Proposition, one of the following two cases must hold:
\begin{itemize}
\item Either the assumption of the Proposition is not fulfilled, i.e., $|A_{i+1}| > 2 |A_i|$, or:
\item $A_1 \subseteq A_1-A_1  \subseteq \Bohr (\Spec_{1-1/576}(A_{i+1}), 1/6)$.
\end{itemize}
Since $A_i$ size is bounded by $\eta$, it suffices to iterate above construction until at most $i \le \log_2 \eta$ to have the latter case satisfied.
Observe that then, by Eq.~\eqref{eq:bohr}, we must then have $\Spec_{1- 1/576}(A_{i+1}) \subseteq B$.
\end{proof}

Observe that if we take any $A_\lambda$ being the set of intersections for an arbitrary $\eta$-labeling of a circulation, then we have already established that $\Bohr(A_{\lambda},1/6) = \{0\}$ (Lemma~\ref{cor:bohr_inter}). By Lemma~\ref{lem:spec_bohr}, we have that for some $\kappa = \bigo(\log \eta)$, $\Spec_{1-1/576}( A_{\lambda} - \kappa A_{\lambda} ) \subseteq \Bohr(A_\lambda, 1/6) = \{0\}$.
Applying Lemma~\ref{lem:spec1} we obtain the following technical claim.
\begin{lemma}
\label{lem:explosion}
Let $A_\lambda$ be the set of intersections for arbitrary $\eta$-labeling of a circulation. Then, $\ka A_{\lambda} = \Z_{\eta}$ for some $\ka = \bigo(\log^2 \eta)$.
\end{lemma}

\section{Bounding the Diameter of Graph \texorpdfstring{$G_\varphi$}{G\_phi}}
\label{sec:circulations}

We now apply the sum-set structure of the set of self-intersections of a circulation, expressed by Lemma~\ref{lem:explosion}, to bound the diameter of graph $G_\varphi$. We start with a direct argument for the case of Eulerian circulations, followed by a more involved analysis for the general case.

\subsection{Eulerian Circulations (Proof of Lemma~\ref{lem:maineuler})}

Considerations of Eulerian circulations rely on the definition of a metric $\delta_t$, which expresses an upper bound on the number of contiguous fragments of the Eulerian circuit which it is sufficient to traverse in order to return to the starting arc with a time offset of $t$.

\paragraph{Intersection distance:}
Fix $\eta = m$. Let us consider a labeling $\lambda$ and the corresponding set of intersections $A_\lambda$. We will say that a function $\delta : \Z \to \Z^+$ is an \emph{intersection distance measure} if it satisfies the following constraints:
\begin{eqnarray}
\label{eq:axiom1}\delta_{i m} = 0&& \text{ for any } i \in \Z\\
\label{eq:axiom2}\delta_x \le 1&& \text{ when } x \in A_{\lambda}\\
\label{eq:axiom3}\delta_{-x} =\delta_x&& \text{ for any } x \in \Z\\
\label{eq:axiom4}\delta_{x+y} \le \delta_{x} + \delta_{y}&& \text{ for any } x,y \in \Z
\end{eqnarray}
From the definition it follows that for any $x \in \kappa A_\lambda$, $\delta_x \le \kappa$. Thus by Lemma~\ref{lem:explosion} we immediately have:
\begin{equation}
\label{eq:maxdelta}
\max_{x \in \Z} \delta_x = \max_{x \in \Z_\eta} \delta_x = \bigo(\log^2 \eta),
\end{equation}
where the first equality holds by property~\eqref{eq:axiom1}, corresponding to shifts by a complete tour of the Eulerian circuit.

Phrasing the above in terms of distances in $G_\varphi$, by considering for any node $v$ the properties of function $\delta_x$ for some arc of the form $(v, \cdot)$, we conclude that the distance in $G_\varphi$ between any two nodes of the form $(v, t)$ and $(v, t+x)$, $x, t\in \Z$ is bounded by $\bigo(\log^2 \eta)$. Since for any pair of vertices $u \neq v$ there exists a path $v \to_{x_{vu}} u$, for some $x_{vu} \in \Z$ we moreover have that pairs of vertices of the form $(v, t)$, $(u, t+x_{vu})$, $t\in \Z$, are at distance $1$ from each other in $G_\varphi$. The claim of Lemma~\ref{lem:maineuler} follows.

\subsection{General Circulations (Proof of Lemma~\ref{lem:maingeneral})}

We now extend the definition of intersection distance so that it is useful for arbitrary circulations.
\paragraph{Generalizing intersection distance:}
Let $\varphi$ be an arbitrary circulation. We will say that a function $\delta : \Z \times E \times E \to \Z^+$ is a \emph{generalized intersection distance measure} if it satisfies the following constraints:
\begin{eqnarray}
\label{eq:axiom5}\delta_0(e,e) = 0&& \quad\text{ for any } e \in E\\
\label{eq:axiom6}\delta_1(e,\varphi(e)) = 0&& \quad\text{ for any } e \in E\\
\label{eq:axiom7}\delta_0(e_1,e_2) = 1&& \quad\text{ when } \pre(e_1)=\pre(e_2)\\
\label{eq:axiom8}\delta_{-x}(e_2,e_1) = \delta_x(e_1,e_2)&& \quad\text{ for any } x \in \Z,e_1,e_2\in E\\
\label{eq:axiom9}\delta_{x+y}(e_1,e_3) \le \delta_{x}(e_1,e_2) + \delta_{y}(e_2,e_3)&& \quad\text{ for any } x,y \in \Z, e_1,e_2,e_3 \in E
\end{eqnarray}

The generalization is to be understood in the following sense: if $\varphi$ is an Euclidean circulation and $\delta$ is a generalized intersection distance measure, then for arbitrary $e \in E$, $\delta_{x} := \delta_{x}(e,e)$ is an intersection distance measure.

\paragraph{Cycle combinatorics:}
We now define a few more concepts regarding how different cycles in a (non-Eulerian) circulation relate to each other.

First, given a circulation with cycles $(E_1, \ldots, E_g)$, we introduce the notion of \emph{adjacency} of cycles, saying that two cycles are \emph{adjacent} if they share a common vertex. Naturally, we can use this definition of adjacency to define a graph $\mathcal G$ on the set of cycles of the circulation, and the corresponding \emph{distance} $\Delta(E_i,E_j)$ between any two cycles $E_i$ and $E_j$ as the number of adjacency hops needed to go between this pair of cycles.

Additionally, given labeling $\lambda$ of the circulation, we say that two cycles $E_i,E_j$ are \emph{$\lambda$-adjacent} if there are two arcs $e_i \in E_i, e_j \in E_j$, such that $\pre(e_i)=\pre(e_j)$ and $\lambda(e_i)=\lambda(e_j)$. We call the graph $\mathcal G_\lambda$ with vertex set corresponding to the set of cycles the cycles $\{E_1,\ldots,E_g\}$ and edges corresponding the $\lambda$-adjacency relation the \emph{cycle graph}. For two cycles $E_i,E_j$, we denote  by $\Delta_\lambda(E_i,E_j)$ the distance in the cycle graph $\mathcal G_\lambda$ between  respective cycles $E_i$ and $E_j$. For compactness of notation, given two arcs $e_1 \in E_1, e_2 \in E_2$, we denote  $\Delta_\lambda(e_1,e_2) = \Delta_\lambda(E_1,E_2)$. Finally, we say that labeling $\lambda$ is \emph{cycle-connected} if its cycle graph $\mathcal G_\lambda$ is connected.

It is easy to see that we can always adjust any indexing $\lambda$ by cyclical shifts, successively for each cycle, so as to incorporate an arbitrary spanning tree of $\mathcal{G}$ in the cycle graph $\mathcal G_\lambda$. In particular, by considering the BFS tree of $\mathcal{G}$, whose diameter is at most $2 \cdot \diam(\mathcal{G})$, we have the following result.

\begin{observation}
\label{obs:diam}
For any $\eta | \gcd(|E_1|,\ldots,|E_g|)$, there is an $\eta$-labeling $\lambda$ such that for any $E_i,E_j$: $\Delta_\lambda(E_i,E_j) \le 2 \cdot \diam(\mathcal{G}).$
\end{observation}

From now on we assume that $\lambda$ satisfies the property given by Observation~\ref{obs:diam}.

\begin{lemma}
\label{lem:delta_alpha}
For any pair of arcs $e_1,e_2$,  there exists $\alpha \in \Z$ such that $\alpha \equiv \lambda(e_2)-\lambda(e_1) \pmod \eta$ and $\delta_{\alpha}(e_1,e_2) \le \Delta_\lambda(e_1,e_2)$. Moreover, there exists $\alpha' \in \Z$, such that $\delta_{\alpha'}(e_1,e_2) \le \Delta(e_1,e_2)$.
\end{lemma}
\begin{proof}
We start by covering the first part of claim.
If $e_1$ and $e_2$ belong to the same cycle, then there is $x \in \Z$ such that $\delta_x(e_1,e_2) = 0$. However, since  $\lambda$ limited to one particular cycle is just a cyclic numbering modulo $n$, we have that $\lambda(e_2)-\lambda(e_1) \equiv x \pmod n$.

We also observe that the claim is satisfied for two arcs that make corresponding cycles $\lambda$-adjacent: if $\lambda(e_1)=\lambda(e_2)$ and $e_1$ and $e_2$ are adjacent, then simultaneously $\lambda(e_2)-\lambda(e_1)=0$ and $\delta_0(e_1,e_2) \le 1$.

By a simple chaining argument we can extend this to any two arcs on $\lambda$-adjacent cycles, and by the induction on $\Delta_\lambda(e_1,e_2)$ to any two arcs.

By repeating the same reasoning using cycle adjacency instead of $\lambda$-adjacency, we get the second claim (without getting any constraint on $\alpha'$).
\end{proof}

\begin{lemma}
\label{lem:cycle_translation}
For any two arcs $e,e'$ and  $x \in \Z$, we have $\delta_x(e',e') \le \delta_x(e,e) + 2\Delta(e,e') \le \delta_x(e,e) +  \bigo(\diam(\mathcal{G}))$.
\end{lemma}
\begin{proof}
There is $\alpha \in \Z$ such that $\delta_\alpha(e,e') \le \Delta(e,e')$. Thus $\delta_x(e',e') \le \delta_\alpha(e',e) + \delta_{x - \alpha+\alpha}(e,e) + \delta_{-\alpha}(e,e') \le  \delta_x(e,e) + 2\Delta(e,e')$.
\end{proof}

\begin{lemma}
\label{th:sqrt}
For any $\delta$ that is a generalized intersection distance measure on a circulation with $g$ orbits, we have:
$$\max_{e \in E, x \in \Z} \delta_x(e,e) = \widetilde\bigo(g).$$
\end{lemma}

\begin{proof}
Let $\eta = \gcd(|E_1|,\ldots,|E_g|)$, and consider an indexing $\lambda$ following Observation~\ref{obs:diam}.
We will once again be working with the set of all intersections $A_\lambda \subseteq  \Z_{\eta}$.

Fix any $x \in A_\lambda$. Thus, there are $e_1 \in E_1, e_2 \in E_2$ such that $\pre(e_1)=\pre(e_2)$ and $\lambda(e_2) - \lambda(e_1) = x$. The cycles $E_1,E_2$ are not necessarily $\lambda$-adjacent, however by Lemma~\ref{lem:delta_alpha} we have that $\delta_{\alpha}(e_1,e_2) = \bigo(\diam(\mathcal{G}))$ for some $\alpha \equiv x \pmod \eta$. Moreover, since $\pre(e_1)=\pre(e_2)$, we have $\delta_{\alpha}(e_1,e_1) \le \delta_{\alpha}(e_1,e_2) + 1$, and so $\delta_{\alpha}(e_1,e_1) = \bigo(\diam(\mathcal{G}))$. By a direct application of Lemma~\ref{lem:cycle_translation} we can extend this bound to any arc $e \in E$: $\delta_{\alpha}(e,e) \le 2 \cdot \Delta(e,e_1) + \delta_{\alpha}(e_1,e_1) = \bigo(\diam(\mathcal{G}))$.
Thus we have that for any $x \in \ka A_\lambda$, there is $x' \in \Z$ such that $\delta_{x'}(e,e) = \bigo(\ka \cdot \diam(\mathcal{G}))$, and $x' \equiv x \pmod \eta$. By Lemma~\ref{cor:bohr_inter} and Lemma~\ref{lem:explosion}, there exists $\ka = \bigo(\log^2 n)$ such that $\ka A_\lambda = \Z_{\eta}$.

Denote $\beta = (x'-x)/\eta$.
We have shown that skewed distances can be bounded for a set which, under the projection $\bmod\ \eta$, is $\Z _{\eta}$. We will now focus on controlling the term $\beta$, describing shifts by multiplies of $\eta$. Pick $e_1 \in E_1$. We have $\delta_{c \cdot |E_1|}(e_1,e_1) = 0$ for any integer $c$, thus by Lemma~\ref{lem:cycle_translation}, $\delta_{c \cdot |E_1|}(e,e) = \bigo(\diam(\mathcal{G}))$ for any arc $e$.

Recall that we put $\eta = \gcd(|E_1|,\ldots,|E_g|)$. We observe that $\eta$ has at most $\bigo(\log \eta) \subseteq \bigo(\log n)$ different primes in its factorization, thus the $\gcd$ is already obtained by taking particular subset of $\log n$ elements. It follows that for any integer $\beta$, there exists a sequence $c_1,\ldots,c_g$ such that $\sum_i c_i |E_i| = \beta \eta$ and there are $\bigo(\log n)$ non-zero coefficients in $c_i$. Thus, in particular, $\delta_{\beta \eta}(e,e) \le \sum_i \delta_{c_i |E_i|}(e,e) = \bigo(\diam(\mathcal{G}) \cdot \log n)$.

Combining shifts by multiplicities of $\eta$ with shifts with proper remainders, we finally have:
$\delta_x(e,e) \le \delta_{x'}(e,e) + \delta_{\beta \eta}(e,e) = \widetilde\bigo(\diam(\mathcal{G}))$, from which the claim follows.
\end{proof}

As in the Eulerian case, we observe that the upper bound on $\delta_x(e,e)$ also gives an upper bound on the path length in graph $G_\varphi$, up to an additive factor of $g$ which allows us to move between two spatially separated vertices. This completes the proof of Lemma~\ref{lem:maingeneral}.

\section{Results for \RR Dynamics}\label{sec:rr}

\subsection{Preliminaries: Properties of \RR and Auxiliary Definitions}

We recall that the RR dynamics was introduced in Section~\ref{sec:introrr}. We define edge load $L_t$ and cumulated load $C_t^{\Delta t}$ as follows: $L_t(e)$ is the number of tokens traversing arc $e$ at timestep $t$, and $C_t^{\Delta t}(e) = \sum \limits_{0\le i<\Delta t} L_{t+i}(e).$

The following notation will be useful when discussing the load-balancing and patrolling properties of the \RR dynamics.

\paragraph{Load balancing:}
A standard notion is to define a discrepancy (over arcs) of a token distribution at a given timestep $t$ as $\max \{\left|L_t(e_1)-L_t(e_2)\right| : e_1,e_2 \in E\}$. Similarly, a notion of a cumulative discrepancy over time $\Delta t$ defined as $\max \{\left|C_t^{\Delta t}(e_1)-C_t^{\Delta t}(e_2)\right| : e_1,e_2 \in E\}$.

\paragraph{Patrolling:}
We say that an arc $e$ has an idleness of $t$ if, starting from some point in time, $e$ is visited at least once in every interval of $t$ consecutive steps, and $t$ is the smallest such value. More precisely, we define \emph{idle arc time} at a given time step $t$ as $\idle(t,e) = \min(\Delta t : C_t^{\Delta t}(e) \ge 1)$. We define idleness for a single arc: $\idle(e) = \limsup \limits_{t \rightarrow \infty} \idle(t,e)$. Finally we define \emph{idleness} as: $\idle(G) = \max\{ \idle(e) : e \in E\}$.

\paragraph{Limit behavior:}

We recall following two results characterizing limit behavior of the \RR system, reached in polynomially many steps.
\begin{theorem}[\!\!\cite{CDGKLU15}, Theorem 11]
\label{th:disc}
\RR dynamics is recurrent, if and only if there is a bijection $\varphi : E \rightarrow E$, such that $\forall_{t \ge 0} \forall_{e \in E}\ L_t(e) = L_{t+1}(\varphi(e))$. Moreover, $\varphi$ can always be selected so that $\pre(\varphi(e)) = \suc(e)$ for any arc $e$.
\end{theorem}

From now on we will identify the subcycle decomposition with bijection $\varphi$. We will denote the Eulerian cycles induced by $\varphi$ as $E_1, E_2, \ldots, E_g$, with $g$ denoting number of cycles.


\begin{theorem}[\!\!\cite{CDGKLU15}, Theorem 26]
\label{th:disc2}
\RR dynamics reaches its recurrent state in $\bigo(m^4\diam^2+m\diam\log k)$ steps on a graph $G$ of diameter $\diam$.
\end{theorem}

By Theorems \ref{th:disc} and \ref{th:disc2} any configuration after polynomially many steps reaches recurrent state. Thus from now on we assume that system at the timestep 0 is already in a recurrent state.

\paragraph{Averaging and number of cycles:}
Let us denote by $k_i$ the total number of tokens circulating cycle $E_i$, that is $k_i = \sum_{e \in E_i} L_t(e)$ (this value does not depend on the choice of $k$). It is easy to see that in order to maintain cumulative discrepancy over arbitrarily long time, the following property has to hold
\begin{equation}
\label{eq:averages}
\frac{k_1}{|E_1|} = \frac{k_2}{|E_2|} = \ldots = \frac{k}{m}.
\end{equation}

As a corollary of \eqref{eq:averages}, the number of cycles $g$ has to satisfy $g \mid \gcd(k,m)$, thus $g \le \gcd(k,m)$ follows. Thus, in particular, if $m$ is prime, then for any $k < m$ we have exactly one cycle. The same occurs whenever $k \perp m$ in the general case.

\subsection{\RR Dynamics with Eulerian Circulation}
Now we proceed to analyze RR dynamics in a recurrent state that has a circulation with a single cycle. Let us fix circulation $\varphi$ as in Theorem~\ref{th:disc}.

\subsubsection{The Similarity Measure}

We use the similarity of the sequence to itself under shifts, where this similarity follows from the small discrepancy of loads on arcs that share a starting (predecessor) vertex.

We define the similarity for time shift $t \in \Z$ as:
\begin{equation}
\label{eq:simplified_distance}
\delta_t = \max_{e \in E, \tau, \Delta t \ge 0} | C_{\tau}^{\Delta t}(e) - C_{\tau+t}^{\Delta t}(e) |.
\end{equation}

\begin{lemma}\label{lem:delta_idm}
$\delta$ defined as in \eqref{eq:simplified_distance} is an intersection distance measure.
\end{lemma}
\begin{proof}
It is immediate to verify that $\delta$ satisfies \eqref{eq:axiom1},\eqref{eq:axiom3} and \eqref{eq:axiom4}. To show~\eqref{eq:axiom2}, take $i \in A_{\lambda}$. By the definition, there are two edges $e_1,e_2$, such that $\pre(e_1) = \pre(e_2)$ and $\lambda(e_2) = \lambda(e_1) + i$. Thus, $e_2 = \varphi^i(e_1)$. Fix $e \in E$. Since $\varphi$ is an Eulerian circulation, let $j \ge 0$ be such that $\varphi^{j}(e) = e_1$. We have, for any $\tau$ and $\Delta t$:
$$| C_{\tau}^{\Delta t}(e) - C_{\tau+i}^{\Delta t}(e) | = | C_{\tau+j}^{\Delta t}(e_1) - C_{\tau+j+i}^{\Delta t}(e_2) | \le 1$$
thus $\delta_i \le 1$ and \eqref{eq:axiom2} is satisfied.
\end{proof}

\subsubsection{Time-Averaging Property}

\begin{lemma}
\label{lem:cum_discrepancy_simplified}
For \RR dynamics on an Eulerian circulation, for any time step $t$ in the recurrent state, $T > 0 $ and $e \in E$, we have $\left|C_{t}^{T}(e) -  \frac{k}{m}\cdot T \right| \le \delta_T \le \bigo(\log^2 \eta)$.
\end{lemma}
\begin{proof}
Fix time interval $T > 0$ and pick $t_1$ and $t_2$, with $t_2 > t_1$, in the recurrent state arbitrarily. We have the following bounds:
$$\left| C_{t_1}^T(e) - C_{t_2}^T(e) \right| = \left| C_{t_1}^{t_2-t_1}(e) - C_{t_1+T}^{t_2-t_1}(e) \right| \le \delta_T,$$
where we used the fact that, for an arbitrary sequence $x_i$ and integers $a \le b \le c \le d$, it holds that: $\sum_{i=a}^{b-1} x_i - \sum_{i=c}^{d-1} x_i = \sum_{i=a}^{c-1} x_i - \sum_{i=b}^{d-1} x_i$, and for the last bound, we used~\eqref{eq:simplified_distance}.

Since $t_1,t_2$ were arbitrary, bounds on the discrepancy of the sequence $(C_{t}^T)_{t \in \Z}$, follow.
Since this sequence is periodic and by equation \eqref{eq:averages} its average value is $T\cdot \frac{k}{m}$, the inequality  $\left|C_{t}^{T}(e) -  \frac{k}{m}\cdot T \right| \le \delta_T$ follows immediately from the previously shown bound. Finally, we have $\delta_T \le \bigo(\log^2 \eta)$ by~\eqref{eq:maxdelta}, since $\delta_T$ is an intersection distance measure following Lemma~\ref{lem:delta_idm}.
\end{proof}

Recalling that by definition $C_{t_0}^{T}(e) = \sum_{t=t_0}^{t_0+T} L_{t}(e)$, we obtain directly from Lemma~\ref{lem:cum_discrepancy_simplified}:
$$
\frac{1}{T}\sum_{t = t_0} ^{t_0+T} L_{t}(e) = \frac{k}{m} \pm \widetilde\bigo\left(\frac{1}{T}\right),
$$
which proves Theorem~\ref{thm:rr} for the special case of Eulerian circulations.




\subsection{\RR Dynamics with General Circulation}

We now apply the results on circulations with multiple cycles to provide bounds for arbitrary \RR dynamics.

\subsubsection{Generalized Similarity Measure}

We extended the notion of $\delta_x$ into a load similarity measure between pair of arcs, under time shift.
Thus, we define skewed distance function, for $t \in \Z$
$$\delta_t( e_1,e_2 ) = \max_{\tau, \Delta t \ge 0} | C_{\tau}^{\Delta t}(e_1) - C_{\tau+t}^{\Delta t}(e_2) |.$$
It is easy to verify that $\delta$ defined as above satisfies \eqref{eq:axiom5},\eqref{eq:axiom6},\eqref{eq:axiom7},\eqref{eq:axiom8}, and \eqref{eq:axiom9}; $\delta$  is thus a generalized intersection distance measure with respect to circulation $\varphi$. We obtain the following formulation of Lemma~\ref{th:sqrt} for the considered function $\delta_t$.

\begin{lemma}
\label{lem:disc_bound_g}
$$\max_{e \in E, t \in \Z} \delta_t(e,e) = \widetilde\bigo(g) \leq \widetilde\bigo(\gcd(k,m)).$$
\end{lemma}

\subsubsection{Time-Averaging Property}

The following Lemma is an automatic generalization of Lemma~\ref{lem:cum_discrepancy_simplified}.

\begin{lemma}
\label{lem:cum_discrepancy}
For \RR dynamics on an arbitrary circulation, for any time step $t$ in the recurrent state, $T > 0 $ and $e \in E$, we have $\left|C_{t}^{T}(e) -  \frac{k}{m}\cdot T \right| \le \delta_T(e,e)$.\qed
\end{lemma}

Recalling once again that by definition $C_{t_0}^{T}(e) = \sum_{t=t_0}^{t_0+T} L_{t}(e)$, Lemmas~\ref{lem:disc_bound_g} and~\ref{lem:cum_discrepancy} directly give:
$$
\frac{1}{T}\sum_{t = t_0} ^{t_0+T} L_{t}(e) = \frac{k}{m} \pm \widetilde\bigo\left(\frac{gcd(k, m)}{T}\right).
$$
which completes the proof of Theorem~\ref{thm:rr} in the general case.

\subsubsection{Other Bounds on the Generalized Similarity Measure}

Going beyond the bound of Theorem~\ref{thm:rr}, we provide a number of corollaries and auxiliary results, upper-bounding the generalized similarity measure $\delta_T(e,e)$ with respect to other graph parameters. We will then use them in conjunction with Lemma~\ref{lem:cum_discrepancy} to bound the cumulated load of the RR dynamics.

Now first provide a $\widetilde\bigo(\sqrt{n})$ bound that uses thresholding: we make use of the fact that in the considered circulation, there are either small cycles (so any cycle must have a short one in its proximity), or there are only large cycles, so there can only be a small number of them. We use the notation developed in Section~\ref{sec:circulations}, denoting distance between cycles as $\Delta(E_i,E_j)$.
\begin{lemma}
\label{th:thresholding1}
For any $e$ and integer $x$, $\delta_x(e,e) = \widetilde\bigo(\sqrt{n})$.
\end{lemma}
\begin{proof}
We will call a cycle that contains at most $\sqrt{n}$ distinct vertices a \emph{small cycle}.

In the adjacency graph of cycles $\mathcal{G}$, any shortest path of length $x$ between two cycles contains a cycle using at most $2n/x$ distinct vertices.

First, let us assume that there is at least one small cycle in the graph.
Thus, for any cycle $E_1$ there is a small cycle $E_2$ such that $\Delta(E_1,E_2) \le \frac{2n}{\sqrt{n}} = \bigo(\sqrt{n})$. Fix $e_1 \in E_1, e_2 \in E_2$. Since $\diam(E_2) \le \sqrt{n}$, by Lemma~\ref{lem:diam}, $\delta_T(e_2,e_2) = \bigo( \sqrt{n} )$ for any integer $T$. By Lemma~\ref{lem:cycle_translation} we then have $\delta_T(e_1,e_1) \le \delta_T(e_2,e_2)+2\Delta(E_1,E_2) = \bigo(\sqrt{n})$.

If there are no small cycles in the graph, then every cycle is of size at least $\sqrt{n}$ distinct vertices. So $\diam(\mathcal{G}) = \bigo(\sqrt{n})$ by the previous observation, and the claim follows from Lemma~\ref{th:sqrt}.
\end{proof}

Next, the following bound trivially holds for small values of $t$, subsequently allowing us to perform analysis of the RR for values of $k$ close to $m$.

\begin{lemma}
\label{lem:small_time}
For any $e$, $\delta_4(e,e) \leq 3$.
\end{lemma}
\begin{proof}
Observe that since arcs $\pre(-e)=\pre(\varphi(e))$, $\delta_1(e,-e) \le 1$. Thus, $\delta_2(e,e) \le 2$. Similarly, one can reason that $\delta_2(\varphi^2(e),e) \le 3$, giving $\delta_4(e,e) \le 3$.
\end{proof}

For our next bound, let us denote by $d(e_1,e_2)$ the distance (length of shortest path) between starting vertices of $e_1$ and $e_2$, measured in $G$. We also denote $\diam(G) = \max\limits_{e_1,e_2 \in E} d(e_1,e_2)$ and $\diam(E_i) = \max\limits_{e_1,e_2 \in E_i} d(e_1,e_2)$. (For sanity of notation, we define $\diam(E_i)$ so that $\diam(E_i) \ge 1$ even if $E_i$ is a loop.) The following Lemma relates cycle diameter with our analysis.
\begin{lemma}
\label{lem:diam}
For any integer $j$ and edge $e \in E_i$, we have: $\delta_{j}(e,e) = \bigo(\diam(E_i)) \leq \bigo(\diam(G))$.
\end{lemma}
\begin{proof}
Assume $j>0$ and fix an integer $r > 0$. Observe that $d(e,\varphi^r(e)) \le \diam(E_i)$, thus we have $\delta_{r+d(e,\varphi^r(e))}(e,e) \le \diam(E_i)+1$. Since for consecutive values of $r$ the value of $r+d(e,\varphi^r(e))$ increases by $0,1$ or $2$, either the claim follows, or $\delta_{j-1}(e,e) \le \diam(E_i)+1$. However, since in every vertex there is a loop, we have $\delta_1(e,e) \le 2$ and the claim follows.
\end{proof}

The following bound on the generalized similarity measure proves useful when considering sufficiently long ranges of time.

\begin{lemma}
\label{cor:sqrtk}
For any integer $B \le m$ and for any $e$, we have: $\min_{x \in [B,2B]} \delta_x(e,e) = \widetilde\bigo(m/B)$.
\end{lemma}
\begin{proof}
First, let us assume that there is at least one cycle with number of arcs not more than $B$. Thus, by a reasoning similar to one from Theorem~\ref{th:thresholding1}, for any cycle $E_1$ and arc $e_1 \in E_1$ there exists a cycle $E_2$ having at most $B$ arcs and that $\Delta(E_1,E_2) \le \frac{m}{B}$. For any $e_2 \in E_2$ there is integer $B \le x \le 2B$ such that $\delta_{x}(e_2,e_2) = 0$  (it suffices for $x$ to be a multiplicity of $|E_2|$). Thus for any $e_1 \in E_1$, we have $\delta_{x}(e_1,e_1) \le 2m/B$. Since $e_1$ was picked arbitrarily, the claim follows.

If every cycle has at least $B$ arcs, then the number of cycles satisfies $g \le m/B$. Thus the claim follows from Lemma~\ref{th:sqrt}.
\end{proof}

Finally, we obtain the following $\bigo(1)$ bound on the generalized similarity measure of the \RR dynamics in trees. We remark that this bound holds even though trees are bipartite.

\begin{lemma}
\label{th:tree}
If $G$ is a tree, then for any $e$ and positive integer $B$, we have: $\min_{x \in [B,2B]} \delta_{x}(e,e) = \bigo(1)$.
\end{lemma}
\begin{proof}

Cycles in $E$ are DFS traversals of some subtree of $G$. We can safely analyse each such subtree separately. Take such subtree $T$, and fix $B \le 2|T|$. (The size here is measured in number of edges.) 
We pick arbitrary root of $T$, and we consider subtrees denoted $T_v$, that is rooted in some vertex $v$ containing everything ``below'' $v$, \emph{plus} the edge connecting $v$ to its parent. One of two cases occurs: (1) there is $v$ such that $2|T_v| \in [B,2B]$, or (2) there is $v$ with children $w_1,w_2,\ldots$ such that $2|T_v| > 2B$ and $2|T_{w_1}|, 2|T_{w_2}|, \ldots < B$.

In the first case, the cycle that takes part in $T_v$ is of size $2|T_v|$, thus we have for some $x \in [B,2B]$ that $\delta_{x}(e',e') = 1$, for some $e'$ in this cycle, thus also for every edge in this cycle.

In the second case, assume that $w_1,w_2,\ldots$ are in the order in which the cycle visits those vertices. Since $|T_{w_1}| + |T_{w_2}| + \ldots = |T_v| - 1 \ge B$, and each of $|T_{w_i}|$ is smaller than $B/2$, we know there prefix sum such that $2|T_{w_1}| + 2|T_{w_2}| + \ldots + 2|T_{w_j}| \in [B,2B]$. Once again we argue that there is a self-intersection of appropriate length, and the claim follows.
\end{proof}

\subsection{Cumulated Load Discrepancy of RR}

Lemma~\ref{lem:cum_discrepancy} provides us with bounds on the cumulated load discrepancy of the RR dynamics in terms of the generalized similarity measure $\delta_T(e,e)$. Introducing Lemmas~\ref{lem:disc_bound_g},~\ref{lem:diam}, and~\ref{th:thresholding1}, we directly obtain the following claims.

\begin{proposition}
\label{pro:cumul}
The cumulated load discrepancy of \RR dynamics in its recurrent state can be upper-bounded as:
\begin{itemize}
\item $\widetilde\bigo(\gcd(k,m))$.
\item $\bigo(\diam(G))$.
\item $\widetilde\bigo(\sqrt{n})$.
\end{itemize}
\qed
\end{proposition}

\subsection{Idleness of RR}

Lemma~\ref{lem:cum_discrepancy} also provides a way of upper-bounding the idleness of the RR dynamics in terms of the generalized similarity measure $\delta_T(e,e)$. Indeed, we have the following claim.

\begin{lemma}
\label{lem:idletime_bound}
If $T > \frac{m}{k} \cdot \delta_T(e,e)$ for some particular value of $T$, then $\idle(e) \le T$.
\end{lemma}
\begin{proof}
By Lemma~\ref{lem:cum_discrepancy}, we obtain $C_t^T(e) \ge \frac{k}{m} \cdot T - \delta_T(e,e) > 0$. Since this bound holds regardless of the choice of $t$, the claim follows.
\end{proof}

The following claims are now obtained directly, by combining Lemma~\ref{lem:idletime_bound} for an appropriately chosen value of $T$ with specific bounds on $\delta_T(e,e)$, given by Lemma~\ref{lem:disc_bound_g} (putting $T=\widetilde{\bigo}(\frac{m}{k} \cdot \gcd(k,m))$), Lemma~\ref{lem:diam} (putting $T=\bigo(\frac{m}{k}\cdot \diam(G))$), Lemma~\ref{th:thresholding1} (putting $T=\widetilde{\bigo}(\frac{m}{k}\cdot\sqrt{n})$), Lemma~\ref{lem:small_time}  (putting $T=4$),~\ref{th:tree}  (putting $T=\arg\min_{x \in [c\cdot m/k, 2c \cdot m/k]} \delta_x(e,e)$ for arbitrary $e$ and $c=\bigo(1)$ being the bound from Lemma~\ref{th:tree}), and Lemma~\ref{cor:sqrtk} (putting $B = \widetilde{\bigo}(m/\sqrt{k})$).

\begin{proposition}
\label{pro:idleness}
For $k \leq m$, the \RR dynamics satisfies the following bounds on idleness:
\begin{itemize}
\item $\idle(G) = \widetilde\bigo(\frac{m}{k} \cdot \gcd(k,m) )$.
\item $\idle(G) = \bigo(\frac{m}{k} \cdot \diam(G) )$.
\item $\idle(G) = \widetilde\bigo(\frac{m}{k} \cdot \sqrt{n})$.
\item $\idle(G) \le 4$ if $k > \frac{3}{4}m$.
\item $\idle(G) = \bigo(\frac{m}{k} )$, on a tree.
\item $\idle(G) = \widetilde\bigo(\frac{m}{k}  \cdot \sqrt{k})$.
\end{itemize}
\qed
\end{proposition}

\section{Conclusion}

The obtained results show that the \RR dynamics achieves almost optimal time-averaging behavior for a vast majority of parameter values. This has implications for the considered patrolling and load-balancing problems, as shown in Table~\ref{tab:results}, most notably in the case when $m$ and $k$ have no large common divisors. The case of an underlying Eulerian circulation has been fully resolved, and the mixing properties of the associated walks along Eulerian circuits are well understood. For the non-Eulerian case, it is an open question whether improved generalized $\delta$-distance bounds can be obtained, independent of the number of cycles $g$ of the considered circulation, with consequent improved bounds on the performance of the \RR dynamics.

\section*{Acknowledgments}
We thank Jurek Czyzowicz and Jukka Suomela for fruitful discussions. We also thank an anonymous referee for valuable comments.

\bibliographystyle{abbrv}
\bibliography{bib}


\appendix
\section{Proof of Proposition~\ref{prop:tao}}

Let $x,y \in A$ and $\xi \in \Spec_{1-\varepsilon}(B - A)$. Then there exists a phase $\theta \in \R / (\eta\mult\Z)$ such that
$$ \Real \sum_{z \in B - A}  e(\xi \cdot z + \theta) \ge (1-\varepsilon)|B-A|,$$
thus
$$\sum_{z \in B - A} (1 - \Real\ e(\xi \cdot z + \theta)) \le \varepsilon |B-A| \le \varepsilon K |B|.$$
Since the summand is non-negative, and $B-A$ contains $b-x$,
$$\sum_{b \in B} |1 - \Real\ e(\xi \cdot (b-x) + \theta)| \le \varepsilon K |B|$$
and by Cauchy-Schwarz
$$\sum_{b \in B} |1 - \Real\ e(\xi \cdot (b-x) + \theta)|^{1/2} \le \varepsilon^{1/2} K^{1/2} |B|.$$
From the elementary identity $|1-e(\alpha)| = \sqrt{2} |1- \Real\ e(\alpha)|^{1/2}$ we conclude that:
$$\sum_{b \in B} |1 - e(\xi \cdot (b-x) + \theta)| \le \sqrt{2} \varepsilon^{1/2} K^{1/2} |B|.$$
A similar expression is obtained for $x$ replaced by $y$. By the triangle inequality we conclude that
$$\sum_{b \in B} | e(\xi \cdot (b-y) + \theta) - e(\xi \cdot (b-x) + \theta)| \le 2 \sqrt{2} \varepsilon^{1/2} K^{1/2} |B|.$$
The left hand side is equal to:
$$\sum_{b \in B} |e(\xi \cdot (b-y) + \theta)| \cdot |e(\xi \cdot (x-y))-1 | = |B| \cdot |e(\xi \cdot (x-y))-1 |,$$
thus we have:
$$|e(\xi \cdot (x-y) )-1| \le \sqrt{8 \varepsilon K}.$$
Since we picked $\xi \in \Spec_{1-\varepsilon}(B-A)$ and $x,y \in A$ arbitrarily, the claim follows.
\end{document}